\documentclass[twocolumn,superscriptaddress,nofootinbib,preprintnumbers,showpacs,tightenlines,notitlepage]{revtex4}
\usepackage[latin1]{inputenc}
\usepackage{graphicx}
\usepackage{amssymb}
\usepackage{color}
\usepackage{float}
\usepackage{amsmath}
\usepackage{amsfonts}
\usepackage{dcolumn}
\usepackage{hyperref}
\usepackage{amsthm}
\usepackage{color}
\usepackage{bm}

\def\uudot{\dot{u}}
\def\3nab{\tilde{\nabla}}

\def\la {\langle}
\def\ra {\rangle}

\def\be {\begin{equation}}
\def\ee {\end{equation}}
\def\ba {\begin{eqnarray}}
\def\ea {\end{eqnarray}}

\newtheorem{thm}{Theorem}

\newcommand{\bra}[1]{\left(#1\right)}

\newcommand{\brac}[1]{\left\{#1\right\}}

\newcommand{\sfr}[2]{{\textstyle\frac{#1}{#2}}}

\newcommand{\lc}{\varepsilon}
\newcommand{\lb}{\{}
\newcommand{\rb}{\}}
\newcommand{\E}{{\mathcal E}}
\renewcommand{\H}{{\mathcal H}}
\newcommand{\barray}{\begin{array}}
\newcommand{\earray}{\end{array}}
\newcommand{\e}{e}
\newcommand{\N}{N}

\newcommand{\del}{\nabla}

 \newcommand{\nab}{\nabla}

\newcommand \ep {\epsilon}

\newcommand \om {\omega}

\newcommand{\udot}{{\mathcal A}}
\newcommand{\hh}{{\mathcal H}}
\begin{document}

\title{New class of LRS spacetimes with simultaneous rotation and spatial twist}
\author{Sayuri Singh}
\email{sayurisingh22@gmail.com }
\affiliation{Astrophysics and Cosmology Research Unit, School of Mathematics, Statistics and Computer Science, University of KwaZulu-Natal, Private Bag X54001, Durban 4000, South Africa.}
\author{George F. R. Ellis}
 \email{george.ellis@uct.ac.za}
 \affiliation{Department of Mathematics and Applied Mathematics and ACGC, University of Cape Town,
Cape Town, South Africa.}
 \author{Rituparno Goswami}
\email{Goswami@ukzn.ac.za}
\affiliation{Astrophysics and Cosmology Research Unit, School of Mathematics, Statistics and Computer Science, University of KwaZulu-Natal, Private Bag X54001, Durban 4000, South Africa.}
\author{Sunil D. Maharaj}
\email{Maharaj@ukzn.ac.za}
\affiliation{Astrophysics and Cosmology Research Unit, School of Mathematics, Statistics and Computer Science, University of KwaZulu-Natal, Private Bag X54001, Durban 4000, South Africa.}

\begin{abstract}
We establish the existence and find the necessary and sufficient conditions for a new class of solutions of Locally Rotationally Symmetric spacetimes that have non vanishing rotation and spatial twist simultaneously. We transparently show that the existence of such solutions demand non vanishing and bounded heat flux and these solutions are self similar. We provide a brief algorithm indicating how to solve the system of field equations with the given Cauchy data on an initial spacelike Cauchy surface. Finally we argue that these solutions can be used as a first approximation from spherical symmetry to study rotating, inhomogeneous, dynamic and radiating astrophysical stars.
 \end{abstract}
 
\maketitle
\section{Introduction}

The spacetimes that are {\em Locally Rotationally Symmetric} (LRS) have been studied in detail and discussed many times in the literature in the cosmological context, i.e. with a fluid matter source  (see for example \cite{Ellis_1967,Ellis_1968,Elst_Ellis_1996} and the references therein). For these spacetimes there exists a continuous isotropy group at each point and hence there is a multiply-transitive isometry group acting on the spacetime manifold. As we know, the isotropies around a point in a spacetime with a fluid can occur as a one-dimensional or three dimensional subgroup of the full group of isometries that leaves the normalised 4-velocity of the matter flow invariant. A three dimensional group of isotropies at each point implies  that the spacetime is isotropic at every point and gives rise to the homogeneous and isotropic Friedmann-Lema\^{i}tre-Robertson-Walker (FLRW) models. While a one dimensional group of isotropies at each point  corresponds to anisotropic and in general spatially inhomogeneous models\footnote{generically with one or two centres where the isotropy group is 3-dimensional}, but includes also some spatially homogeneous (Bianchi and Kantowski-Sachs) models \cite{Mac_Ellis_1969, King_Ellis_1973}. The important property of LRS spacetimes is that they exhibit locally (at each point) a unique preferred spatial direction, covariantly defined, (for example, by a vorticity vector field, a non-vanishing non-gravitational acceleration of the matter, or a density gradient). 

LRS spacetimes with a perfect fluid matter source have been completely analysed and classified by Stewart and Ellis using tetrad methods \cite{Ellis_1968}. 
 Using a semi-tetrad covariant formalism it was shown that the Einstein field equations can be written as a set of first order equations of geometrical scalars \cite{Elst_Ellis_1996, Clarkson:2007yp}. By analysing the consistency conditions of the field equations, it was rigorously proved that a perfect fluid LRS spacetime cannot have a simultaneous fluid rotation and spatial twist of the preferred spatial direction. Based on this observation, the perfect fluid LRS spacetimes can be divided into three distinct classes. {\em Class I} spacetimes are those where the rotation is non zero but the twist vanishes. This class was shown to be non-expanding, non-distorting and stationary and the solutions generalise the well known G{\"o}del solution. In {\em Class II} spacetimes both the rotation and the twist vanish and these consist of the spherical, hyper-spherical, and plane symmetric (cylindrical) solutions. {\em Class III} spacetimes have no rotation or acceleration but non-zero twist of the preferred spatial direction. These spacetimes are spatially homogeneous.

Though all these classes are of 
 interest, and LRS-II solutions have been used extensively to study spherically symmetric astrophysical objects, none of them are suitable for modelling a dynamical rotating star (gravitational collapse of a rotating star, for example). For LRS-I, the rotation is non-zero but the spacetime is stationary, while the other two classes allow dynamical solutions with vanishing rotation. In this study the three key questions are: {\em By relaxing the condition of a perfect fluid, that is by introducing pressure anisotropy and heat flux, is it possible to have dynamical solutions with non-zero rotation and non-zero twist? If yes, can these solutions be physical? What are the local geometrical properties  of such solutions?}

In this paper we investigate in detail the above questions by using the semi-tetrad 1+1+2 covariant formalism \cite{Clarkson:2002jz,Betschart:2004uu,Clarkson:2007yp}. We first establish the existence of such solutions and then find the constraints on the thermodynamic quantities of matter that generate such solutions. We also demonstrate that there exist physically realistic solutions where the matter satisfies physically reasonable energy conditions.

The paper is organised as follows: In the next two sections we describe briefly the basic concepts of local semi-tetrad 1+3 and 1+1+2 covariant formalisms. In the subsequent sections we discuss the properties of LRS spacetimes and the field equations written in terms of the 1+1+2 geometrical variables. In section 5, we then proceed to show the existence of dynamic solutions for imperfect fluids (with pressure anisotropy and heat flux) with non-zero rotation and spatial twist. We also investigate transparently the constraints that the thermodynamic quantities of the matter must satisfy for such solutions to exist. We provide a brief algorithm indicating how to solve the system of field equations with the given initial data. Finally we briefly discuss about how these solutions can be used as a first approximation to spherical symmetry in order to study rotating, inhomogeneous and dynamic astrophysical objects.

Unless otherwise specified, we use natural units ($c=8\pi G=1$)  and $(-,+,+,+)$ signature throughout this paper.  
The symbol $\nabla$ represents the usual covariant derivative. The Riemann tensor is defined by
\begin{equation}
R^{a}{}_{bcd}=\Gamma^a{}_{bd,c}-\Gamma^a{}_{bc,d}+ \Gamma^e{}_{bd}\Gamma^a{}_{ce}-\Gamma^e{}_{bc}\Gamma^a{}_{de}\;,
\end{equation}
and the Ricci tensor is obtained by contracting the {\em first} and {\em third} indices
\begin{equation}\label{Ricci}
R_{ab}=g^{cd}R_{cadb}\;.
\end{equation}
The Hilbert--Einstein action in the presence of matter is given by
\begin{equation}
{\cal S}=\frac12\int d^4x \sqrt{-g}\left[R-2\Lambda-2{\cal L}_m \right]\;,
\end{equation}
variation of which gives the Einstein field equations as
\be
G_{ab}+\Lambda g_{ab}=T_{ab}\;
\ee
where $G_{ab} := R_{ab} - \frac{1}{2}R g_{ab}$, $R :=R^a_{\,\,\,a}$, and $\Lambda$ is the cosmological constant.

\section{1+3 decomposition of spacetime}

With respect to a timelike congruence, the spacetime can be locally decomposed into time and space parts. One natural way to define such a timelike congruence would be along the matter flow, with the \textit{four-velocity} defined as
\be
u^a = \frac{dx^a}{d\tau}, \quad \mbox{with} \quad u^a u_a = -1 ,
\ee
where $\tau$ is the proper time. Given the \textit{four-velocity} $u^a$, we have the unique \textit{projection tensors}
\ba
U^a{}_b &=& -u^a u_b ,
\\
h^a{}_b &=& g^a{}_b+u^au_b,
\ea
where $h^a{}_b$ is the projection tensor that projects any 4D vector or tensor onto the local 3-space orthogonal to $u^a$. It follows that
\begin{align*}
U^a{}_c U^c{}_b &= -U^a{}_b &  U^a{}_b u^b &= u^a, & U^a{}_a &=1, \\
h_{ab}u^b &= 0 & h^a{}_c h^c{}_b &= h^a{}_b, & h^a{}_a &=3.
\end{align*}

With the choice of this timelike vector, we have two well defined directional derivatives. We have the vector $u^{a}$ which is used to define the \textit{covariant time derivative} (denoted by a dot) for any tensor $ S^{a..b}{}_{c..d}$, given by 
\be
\dot{S}^{a..b}{}_{c..d}{} = u^{e} \nab_{e} {S}^{a..b}{}_{c..d} 
\ee
and we have the tensor $h_{ab}$ which is used to define the fully orthogonally \textit{projected covariant derivative} $D$ for any tensor $
S^{a..b}{}_{c..d} $: 
\be D_{e}S^{a..b}{}_{c..d}{} = h^a{}_f
h^p{}_c...h^b{}_g h^q{}_d h^r{}_e \nab_{r} {S}^{f..g}{}_{p..q}\;,
\ee 
with total projection on all the free indices. 
The splitting of the spacetime gives a 3-volume element
\be
\ep_{abc}= \eta_{abcd}u^d, \; \mbox{where} \; \ep_{abc} = \ep_{[abc]} \; \mbox{and} \; \ep_{abc}u^c=0.
\ee
Since $\eta_{abcd}$ is the four-dimensional volume element, i.e., $\eta_{abcd}=\sqrt{|\mbox{det}\,g|}\delta^0_{\left[ a
\right. }\delta^1_b\delta^2_c\delta^3_{\left. d \right] }$, we have 
\be
\eta_{abcd} = 2u_{\left[ a \right.}\ep_{\left. bcd\right]}.
\ee
Since $\eta_{abcd}$ is skew-symmetric, the following contractions hold
\ba
\ep_{abc} \ep^{def} &=& 3!h^d{}_{\left[ a \right. }h^e{}_bh^f{}_{\left. c \right] }, \\
\ep_{abc} \ep^{dec} &=& 2h^d{}_{\left[ a \right.}h^e{}_{\left. b \right]}, \\
\ep_{abc} \ep^{dbc} &=& 2h^d{}_a, \\
\ep_{abc} \ep^{abc} &=& 3.
\ea

The covariant derivative of $u^a$ can be decomposed as
\be
\nabla_a u_b = -u_a A_b + D_a u_b,
\ee
where $D_a$ totally projects derivatives onto the 3-space. $D_a u_b$ can be decomposed into the trace part, the trace-free symmetric part and the trace-free anti-symmetric part, i.e.,
\be
\nabla_a u_b = -u_a A_b +\frac13\Theta h_{ab}+\sigma_{ab}+\ep_{abc}\om^{c},
\ee
where $A_b=\dot u_b$ is the acceleration, $\Theta=D_au^a$ represents the expansion of $u_a$, 
$\sigma_{ab}=\bra{h^c{}_{\left( a \right.}h^d{}_{\left. b \right)}-\sfr13 h_{ab} h^{cd}}D_cu_d$ is the shear tensor that denotes the distortion and $\om^{c}$ is the vorticity vector denoting the rotation.

The \textit{Weyl curvature tensor} $C_{abcd}$, which gives the locally free gravitational field, is defined by the equation
\be
C^{ab}{}_{cd} := R^{ab}{}_{cd}-2g^{\left[ a \right.}{}_{\left[c \right.}R^{\left. b\right]}{}_{\left. d\right]}+\sfr13Rg^{\left[ a \right.}{}_{\left[c \right.}g^{\left. b\right]}{}_{\left. d\right]}.
\ee
Since the Weyl tensor is trace-free on all its indices $\bra{C^c{}_{acb}=0}$, the Ricci tensor $R_{ab}$ is the trace of $R_{abcd}$, and $C_{abcd}$ is the trace-free part. The Weyl tensor can be split relative to $u^a$ into the \textit{electric} and \textit{magnetic Weyl curvature} parts as
\begin{eqnarray}
E_{ab} &=& C_{abcd}u^bu^d  \\
&\Rightarrow&  E^a{}_{a} = 0,~~E_{ab} = E_{\la ab \ra},~~E_{ab}u^b = 0,
\end{eqnarray}
and
\begin{eqnarray}
H_{ab} &=& \sfr12\ep_{ade}C^{de}{}_{bc}u^c  
\\
&\Rightarrow&  H^a{}_{a} = 0,~~ H_{ab} = H_{\la ab \ra},~~ H_{ab}u^b = 0.
\end{eqnarray}
The energy momentum tensor of matter can be decomposed similarly as 
\be
T_{ab}=\mu u_au_b+q_au_b+q_bu_a+ph_{ab}+\pi_{ab}\;,
\ee
where $p=(1/3 )h^{ab}T_{ab}$ is the isotropic pressure, $\mu=T_{ab}u^au^b$ is the energy density, $q_a=q_{\la a\ra}=-h^{c}{}_aT_{cd}u^d$ is the 3-vector that defines the heat flux, and $\pi_{ab}=\pi_{\la ab\ra}$ is the anisotropic stress.

\section{1+1+2 decomposition of spacetime}

The 1+1+2  decomposition is a natural extension of 1+3 decomposition, where with respect to a given spatial direction the 3-space is further decomposed, that is we now have another split along a preferred spatial direction. We choose a spacelike vector field $e^a$  such that
\be
u^a e_a =0 \quad \mbox{and} \quad e^a e_a=1.
\ee

Then the new projection tensor is given by
\be \label{N1}N_a^{~b}\equiv
h_a^{~b}-\e_a\e^b=g_{a}^{~b}+u_au^b-\e_a\e^b\, . \ee
This tensor projects vectors onto local 2-spaces, defined as \textit{sheets} (note that these are not subspaces of the 3-space if the twist of $e^a$ is nonzero). Thus
\be\label{N2}
e^a N_{ab}=0=u^a N_{ab}, \quad N^a{}_a=2.
\ee
The volume element of this sheet is
\be \lc_{ab}\equiv\ep_{abc}\e^c = u^d\eta_{dabc}e^c\;\ee
Using the definitions of  $\lc_{ab}$ and $N_{ab}$, we have the following conditions
\ba
 \lc_{ab}\e^b &=&0=\lc_{(ab)}\,,
\\
 \lc_{abc} &=&  e_a \lc_{bc} - e_b  \lc_{ac} + e_c  \lc_{ab}\,,
\\
 \lc_{ab} \lc^{cd} &=& N_a{}^c N_b{}^d - N_s{}^d N_b{}^c ,
\\
 \lc_a{}^c \lc_{bc} &=& N_{ab} ,
\\
 \lc^{ab} \lc_{ab} &=& 2.
\ea
Any 3-vector $\psi^a$ can now be irreducibly split into a scalar, $\Psi$, which
is the vector component parallel to $\e^a$, and a vector,
$\Psi^a$ that lies in the sheet as follows:
\ba 
\psi^a&=&\Psi\e^a+\Psi^{a},~~~\mbox{where}~~~\Psi\equiv \psi_a\e^a\
,\nonumber\\&&~~~\mbox{and}~~~\Psi^{a}\equiv \N^{ab}\psi_b\equiv \psi^{\bar a}\label{psia},
\ea where the bar over the index denotes projection with $\N_{ab}$.
Similarly, the same can be done for any 3-tensor, $\psi_{ab}$,
\be
\psi_{ab}=\psi_{\langle
ab\rangle}=\Psi\bra{\e_a\e_b-\sfr12\N_{ab}}+2\Psi_{(a}\e_{b)}+\Psi_{{ab}}\
, \label{tensor-decomp} 
\ee where
 \ba
\Psi&\equiv &\e^a\e^b\psi_{ab}=-\N^{ab}\psi_{ab}\ ,\nonumber\\
\Psi_a&\equiv &\N_a^{~b}\e^c\psi_{bc}=\Psi_{\bar a}\ ,\nonumber\\
\Psi_{ab}&\equiv &
\bra{\N_{(a}^{~~c}\N_{b)}^{~~d}-\sfr{1}{2}\N_{ab}\N^{cd}}\psi_{cd}
\equiv \Psi_{\lb ab\rb}\ . \label{PSTF-TT}
\ea 
Apart from the `{\it time}' (dot) derivative, we introduce two new derivatives, which for any tensor $ \psi_{a...b}{}^{c...d}  $: 
\ba
\hat{\psi}_{a..b}{}^{c..d} &\equiv & e^{f}D_{f}\psi_{a..b}{}^{c..d}~, 
\\
\delta_f\psi_{a..b}{}^{c..d} &\equiv & N_{a}{}^{p}...N_{b}{}^gN_{h}{}^{c}..
N_{i}{}^{d}N_f{}^jD_j\psi_{p..g}{}^{i..j}\;.
\ea 
The derivative along the $e^a$ vector-field in the surfaces orthogonal to $ u^{a}$ is called the hat-derivative, while the derivative projected onto the sheet is called the $\delta$ -derivative. This projection is on every free index.
 
Using (\ref{psia}) and (\ref{tensor-decomp}), the usual 1+3 kinematical and Weyl quantities can now be split into the irreducible set \[{\cal D}_1 :=  \{\Theta,\udot,\Omega,\Sigma,{\cal E},{\cal H},\udot^a,\Sigma^a,{\cal
E}^a,{\cal H}^a,\Sigma_{ab},{\cal E}_{ab},{\cal H}_{ab}\}. \] The 4-acceleration,vorticity and shear split as
\ba
\uudot^a&=&\udot \e^a+\udot^a,\\
\omega^a&=&\Omega \e^a+\Omega^a,\\
\sigma_{ab}&=&\Sigma\bra{\e_a\e_b-\sfr{1}{2}\N_{ab}}+2\Sigma_{(a}\e_{b)}+\Sigma_{ab}.
\ea
For the electric and magnetic Weyl tensors we get
\ba
E_{ab}&=&{\cal E}\bra{\e_a\e_b-\sfr{1}{2}\N_{ab}}+2{\cal E}_{(a}\e_{b)}+{\cal E}_{ab},\\
H_{ab}&=&{\cal H}\bra{\e_a\e_b-\sfr{1}{2}\N_{ab}}+2{\cal H}_{(a}\e_{b)}+{\cal
H}_{ab}.
\ea
Similarly the fluid variables, $q^a$ and $\pi_{ab}$, may be split as follows
\ba
q^a&=&Q \e^a+Q^a,\\
\pi_{ab}&=&\Pi\bra{\e_a\e_b-\sfr{1}{2}\N_{ab}}+2\Pi_{(a}\e_{b)}+\Pi_{ab}.
\ea

By decomposing the covariant derivative of $e^a$ in the direction orthogonal to $u^a$ into it's irreducible parts, i.e., the spatial derivative of $e^a$, we get
\be 
 {D}_{a}e_{b} = e_{a}a_{b} + \frac{1}{2}\phi N_{ab} + 
\xi\epsilon_{ab} + \zeta_{ab}~, 
\ee
where 
\ba 
a_{a} &\equiv & e^{c}{\rm D}_{c}e_{a} = \hat{e}_{a}~, \\ 
\phi &\equiv & \delta_ae^a~, \\  \xi &\equiv & \frac{1}{2} 
\epsilon^{ab}\delta_{a}e_{b}~, \\ 
\zeta_{ab} &\equiv & \delta_{\lb a}e_{b \rb }~.
\ea
Here, $\phi$ represents the \textit{expansion of the sheet},  $\zeta_{ab}$ is the \textit{shear}, i.e., the distortion of the sheet, $a^{a}$ its \textit{acceleration} and $\xi$ is its (spatial) \textit{vorticity}, i.e., the ``twisting'' or rotation of the sheet. 

\section{LRS spacetime field equations}

As discussed in the Introduction, in LRS spacetime there exists a unique, preferred spatial direction at each point and this preferred direction is covariantly defined.  This direction creates a local axis of symmetry, i.e., all observations are identical under rotations about it. In particular, they are the same in all spatial directions that are perpendicular to that direction. 
Hence the 1+1+2 decomposition described in the previous section is ideally suited for the study of LRS spacetimes. We can immediately see that if we choose the spacelike unit vector $e^a$ along the preferred spatial direction of the spacetime, then by symmetry all the sheet vectors and tensors vanish identically. 
 Thus, all the non-zero 1+1+2 variables are covariantly defined scalars. The geometrical scalar variables that fully describe LRS spacetimes are \[ {\cal D}_2 := \brac{\udot, \Theta,\phi, \xi, \Sigma,\Omega, \E, \H, \mu, p, \Pi, Q }. \] Decomposing the Ricci identities for $u^a$ and $e^a$ and the doubly contracted Bianchi identities, we now get the following field equations for LRS spacetimes:
\smallskip

\textit{Evolution}:
\ba
   \dot\phi &=& \bra{\sfr23\Theta-\Sigma}\bra{\udot-\sfr12\phi}
+2\xi\Omega+Q\ , \label{phidot}
\\ 
\dot\xi &=& \bra{\sfr12\Sigma-\sfr13\Theta}\xi+\bra{\udot-\sfr12\phi}\Omega
\nonumber \\ && +\sfr12 \hh,  \label{xidot}
\\
\dot\Omega &=& \udot\xi+\Omega\bra{\Sigma-\sfr23\Theta}, \label{dotomega}
\\
\dot \hh &=& -3\xi\E+\bra{\sfr32\Sigma-\Theta}\hh+\Omega Q
\nonumber\\ && +\sfr32\xi\Pi.
\ea
\smallskip

\textit{Propagation}:
\ba
\hat\phi  &=&-\sfr12\phi^2+\bra{\sfr13\Theta+\Sigma}\bra{\sfr23\Theta-\Sigma}
    \nonumber\\&&+2\xi^2-\sfr23\bra{\mu+\Lambda}
    -\E -\sfr12\Pi,\,\label{hatphinl}
\\
\hat\xi &=&-\phi\xi+\bra{\sfr13\Theta+\Sigma}\Omega , \label{xihat}
\\
\hat\Sigma-\sfr23\hat\Theta&=&-\sfr32\phi\Sigma-2\xi\Omega-Q\
,\label{Sigthetahat}
 \\
  \hat\Omega &=& \bra{\udot-\phi}\Omega, \label{Omegahat}
\\
\hat\E-\sfr13\hat\mu+\sfr12\hat\Pi&=&
    -\sfr32\phi\bra{\E+\sfr12\Pi}+3\Omega\hh
 \nonumber\\&&   +\bra{\sfr12\Sigma-\sfr13\Theta}Q , \label{Ehatmupi}
\\
\hat \hh &=& -\bra{3\E+\mu+p-\sfr12\Pi}\Omega
\nonumber\\&&-\sfr32\phi \hh-Q\xi,
\ea
\smallskip

\textit{Propagation/evolution}:
\ba
   \hat\udot-\dot\Theta&=&-\bra{\udot+\phi}\udot+\sfr13\Theta^2
    +\sfr32\Sigma^2 \nonumber\\
    &&-2\Omega^2+\sfr12\bra{\mu+3p-2\Lambda}\ ,\label{Raychaudhuri}
\\
    \dot\mu+\hat Q&=&-\Theta\bra{\mu+p}-\bra{\phi+2\udot}Q \nonumber \\
&&- \sfr32\Sigma\Pi,\,
\\    \label{Qhat}
\dot Q+\hat
p+\hat\Pi&=&-\bra{\sfr32\phi+\udot}\Pi-\bra{\sfr43\Theta+\Sigma} Q\nonumber\\
    &&-\bra{\mu+p}\udot\ ,
\ea\\
\ba
\dot\Sigma-\sfr23\hat\udot
&=&
\sfr13\bra{2\udot-\phi}\udot-\bra{\sfr23\Theta+\sfr12\Sigma}\Sigma\nonumber\\
        &&-\sfr23\Omega^2-\E+\sfr12\Pi\, ,\label{Sigthetadot}
\\  
\dot\E +\sfr12\dot\Pi +\sfr13\hat Q&=&
    +\bra{\sfr32\Sigma-\Theta}\E-\sfr12\bra{\mu+p}\Sigma \nonumber \\
  && -\sfr12\bra{\sfr13\Theta+\sfr12\Sigma}\Pi+3\xi\hh \nonumber\\
    &&+\sfr13\bra{\sfr12\phi-2\udot}Q
\label{edot}.
\ea

\textit{Constraint:}
\be
\hh = 3\xi\Sigma-\bra{2\udot-\phi}\Omega. \label{H}
\ee

Also we give the commutation relation for the {\it dot} amd {\it hat} derivatives, for LRS spacetimes: 
\be
\hat {\dot \Psi}-\dot {\hat \Psi} = -\udot \dot\Psi+\bra{\sfr13\Theta+\Sigma} \hat\Psi, \label{psihatdot}
\ee
which holds true for any scalar $\Psi$.
\\

\textbf{Perfect fluids:} Now, if we consider a perfect fluid with $Q=\Pi=0$, the consistency conditions of LRS spacetimes (which are derived from the fact that the propagation equations should be identically evolved in time) demand that \begin{equation}
\Omega\xi=0.
\end{equation} Due to this condition, the spacetime is divided into 3 distinct subclasses \cite{Ellis_1968,Elst_Ellis_1996}: 
\begin{enumerate}
\item\textbf{LRS class I (Rotating solutions):} $\Omega\neq 0$ \\
$e^b$ is hypersurface orthogonal and $u^b$ is twisting. When $\Omega\neq 0$, we see that $\Rightarrow \xi=0=\Theta=\Sigma, \, \dot f=0$. Therefore models with LRS class I solutions can neither expand nor distort. These models are stationary as the dot of all the scalar quantities vanish.
 
\item\textbf{LRS class II (The Inhomogeneous orthogonal family):} $\xi=0=\Omega$ \\
Here, \textit{both} $e^b$ and $u^b$ are hypersurface orthogonal. When $\xi=0=\Omega$, there exist 3-surfaces orthogonal to the fluid flow. All models in this dynamic and spatially inhomogeneous LRS class have vanishing `magnetic part' of the {Weyl curvature tensor}.
 
\item\textbf{LRS class III (Homogeneous orthogonal models with twist):} $\xi\neq 0$ \\
$e^b$ is twisting and $u^b$ is hypersurface orthogonal. When $\xi\neq 0 \Rightarrow \Omega=\phi=\udot=0$, \textit{all} spatial derivatives vanish and all scalars \textit{f} are spatially homogeneous. We see that the spacetimes themselves are orthogonally spatially homogeneous (OSH). 
\end{enumerate}
\section{A new class of LRS with $\Omega, \xi \neq 0$}

As discussed earlier we would now like to relax the perfect fluid condition, that is we introduce pressure anisotropy and heat flux in the matter, and look for existence of solutions that have both rotation and twist of the preferred direction. To do this, first let us derive an important result for LRS spacetimes. 
We can write the full covariant derivatives of the vectors $u^a$ and $e^a$ in terms of the LRS scalars in the following way:
\ba
\del_au_b&=&-\udot u_ae_b+\bra{\sfr13\Theta+\Sigma}e_ae_b \nonumber
\\&&+\bra{\sfr13\Theta-\sfr12\Sigma}N_{ab}+\Omega\lc_{ab}, \label{delaub}
\ea
\ba
\del_ae_b &=& -\udot u_au_b +\bra{\sfr13\Theta+\Sigma}e_au_b \nonumber
\\ && +\sfr12\phi N_{ab}+\xi\lc_{ab}. \label{delaeb}
\ea
Contracting the above with $\lc_{ab}$ and using 
\be
 \lc^{ab}u_a =\lc^{ab}e_a=0; \; \lc^{ab}\lc_{ab}=2, \label{epabua},
\ee
we get  
\be
\Omega = \sfr12\lc^{ab}\del_au_b, \quad \xi = \sfr12\lc^{ab}\del_ae_b . \label{omegaxi}
\ee
Now, for any scalar function `$\Psi$', we have
\be
\del_b \Psi= -\dot\Psi u_b+ \hat\Psi e_b.
\ee
Differentiating again we have
\ba
\del_a \del_b \Psi&=& -\bra{\del_a \dot\Psi}u_b - \dot \Psi\bra{\del_a u_b} \nonumber
\\&&+\bra{\del_a \hat \Psi}e_b+\hat \Psi \bra{\del_a e_b} . \label{deladelbf}
\ea
Contracting with $\lc^{ab}$, and noting that $\del_a\del_b \Psi$ is symmetric in \textit{a} and \textit{b}, we see that the LHS of (\ref{deladelbf}) vanishes. Using equations (\ref{epabua}) and (\ref{omegaxi}) we get an important result:
\be\label{scalarcons}
\forall \Psi, \,\,\, \dot\Psi\Omega = \hat\Psi \xi. 
\ee
This equation implies self-similarity, for it applies to all scalars, and is unchanged under the transformation $t \rightarrow at,$ $r \rightarrow ar,$ where $t$ and $r$ are the curve parameters of the integral curves of $u$ and $e$.\\

 From the above equation it is clear that if $\Omega\ne0, \xi=0$, the dot derivatives of all the scalars vanish, making the spacetime stationary. On the other hand if $\Omega=0, \xi\ne0$, the hat derivatives of all scalars vanish, making the spacetime spatially homogeneous. Thus we arrive to an important result:
 \begin{thm}
 For LRS spacetimes with non-zero rotation and spatial twist, there always exists a conformal Killing vector in the $[u,e]$ plane. When one of these quantities vanish then the conformal Killing vector becomes a Killing vector. This Killing vector is timelike for vanishing spatial twist and it is spacelike for vanishing rotation. However when both the rotation and spatial twist vanish no such symmetry is guaranteed. 
 \end{thm}
 
 Another important point to be noted here is that $\Omega$ and $\xi$ do not evolve independently. Supposing that both are not equal to zero, first considering $\Psi=\Omega$ in equation   
(\ref{scalarcons}) and then $\Psi=\xi$ and using the field equations (\ref{dotomega}), (\ref{dotomega}), (\ref{dotomega}), (\ref{xihat}) and (\ref{H}) we get the constraint
\be
\frac{\Omega}{\xi}=-\frac{\phi}{\Sigma-\sfr23\Theta} \,.
\label{constraint1}
\ee

Now to establish the existence of solutions with non-zero rotation and spatial twist, we state and prove the following theorem:
\begin{thm}
Evolution of all the independent geometrical scalars of LRS spacetimes that have non-zero rotation and spatial twist, obey a common second order linear hyperbolic partial differential equation and the existence of a initial spacelike Cauchy surface is guaranteed. Subject to the initial Cauchy data on this surface these geometrical scalars can be uniquely determined, and hence unique solutions of the field equations exist.
\end{thm}
\begin{proof}
Taking the {\it dot} derivative and then the {\it hat} derivative of the equation (\ref{scalarcons}), subtracting them and then using the commutation relation (\ref{psihatdot}) and the  field equations (\ref{dotomega}), (\ref{dotomega}), (\ref{dotomega}), (\ref{xihat}) and (\ref{H}), we obtain the following equation $\forall \Psi,$
\ba\label{waveeqn}
 -\Omega^2\ddot\Psi+\xi^2\hat{\hat{\Psi}}&-&\dot\Psi\Omega\left[\xi(\udot-\phi)+\Omega(\Sigma-\sfr23\Theta)\right]\nonumber\\
&+&\hat{\Psi}\xi\left[2\Omega\Sigma-\sfr13\Omega\Theta-\phi\xi\right]=0.
\ea
We can easily see that the above equation is a hyperbolic (wave like) second order linear partial differential equation  for $\Omega,\xi\ne0$, that governs the evolution of all independent geometrical scalars that describe a LRS spacetime. By the properties of hyperbolic partial differential equations, there exists a unique solution subject to Cauchy initial data on a spacelike Cauchy surface. To check whether such a 3 dimensional spacelike surface exist, let us consider the Lie derivative of the tensor $N^{ab}$ with respect to the spacelike vector $e^a$. We know that 
\be 
(\mathcal{L}_{\bm{e}}N)^{ab}=e^c\nab_cN^{ab}-N^{cb}\nab_ce^a-N^{ac}\nab_ce^b\;.
\ee
Using (\ref{N1}),(\ref{N2}),(\ref{delaub}),(\ref{delaeb}) and (\ref{epabua}) we see that
\be
(\mathcal{L}_{\bm{e}}N)^{ab}=\phi N^{ab}\;,
\ee
which implies
\be
(\mathcal{L}_{\bm{e}}N)^{ab}u_a=(\mathcal{L}_{\bm{e}}N)^{ab}u_b=0\;,
\ee
that is neither the vector $e^a$ and the tensor $N^{ab}$, nor the Lie derivative of $N^{ab}$ with respect to $e^a$ has any component along the timelike vector $u^a$. 
This clearly shows that the tensor product of $e^a$ and $N^{ab}$ indeed spans
a spacelike 3-surface  where we can specify the Cauchy initial data to obtain a unique solution of (\ref{waveeqn}) for all the independent geometrical and thermodynamic scalars of the LRS spacetime. 
\end{proof}
Furthermore, the hyperbolic nature of the above equation dictates the existence of two families of characteristics. In analogy with the incoming and outgoing waves, these characteristics describe the expanding/collapsing branches of the solutions.

\subsection{Constraints on thermodynamic variables}

We will now describe the constraints on the thermodynamic variables for the energy momentum tensor of the matter field, that generates LRS solutions with non-zero rotation and spatial twist. We first observe that the common wave like equation (\ref{waveeqn}) was obtained by the Ricci identities of the timelike vector $u^a$ and spacelike vector $e^a$. To obtain the constraints on the matter variables, we need to look at the consistencies of doubly contracted Bianchi identities carefully. 
We state and prove the following theorem here:
\begin{thm}
The necessary condition for a LRS spacetime to have non-zero rotation and spatial twist simultaneously is non-zero heat flux which is bounded from both sides.
\end{thm}
\begin{proof}
Taking the \textit{time-like} derivative for the equation (\ref{Ehatmupi}) and using (\ref{psihatdot}) and the field equations, we get
\be
\Omega\xi\bra{\mu+p+\Pi}+Q\bra{\Omega^2+\xi^2}=0. \nonumber
\ee
Simplifying the above equation we get
\be
\sfr{\sfr\Omega\xi}{1+\bra{\sfr\Omega\xi}^2} = \frac{-Q}{\mu+p+\Pi}.\label{omegaxi2}
\ee
From the above equation it is clear that if we demand both $\Omega$ and $\xi$ are well defined and non-zero, and all the energy conditions to be satisfied we must have $Q \neq 0$. Also it is interesting to note the the ratio of the rotation and spatial twist can be described in terms of the thermodynamic quantities only.
 Now using (\ref{omegaxi2}) to solve for $\sfr\Omega\xi$ gives:
\be\label{constraint2}
\frac{\Omega}{\xi} = \frac{-\bra{\mu+p+\Pi}\mp \sqrt{\bra{\mu+p+\Pi}^2-4Q^2}}{2Q}.
\ee
From the above equation it is clear that for the rotation and spatial twist to be well defined, real and non-zero, we must have $(\mu+p+\Pi)^2  > 4Q^2$. Thus the thermodynamic quantities must satisfy the following constraint
\be
-\frac12(\mu+p+\Pi) <Q<\frac12(\mu+p+\Pi)\;;\; Q \neq 0 .\label{ConsCon}
\ee
\end{proof}
From \cite{santos,chan} we can see that the above constraints are consistent with the Dominant energy conditions (DEC) for the matter field. Hence we do have matter that obeys the physically reasonable energy conditions that can generate a LRS spacetime with non-zero rotation and spatial twist. Also the rest of the propagation equations evolve identically in time and give no new constraints.

\subsection{Other constraints and solution finding algorithm}

Let us now try to reduce the number of independent geometrical scalars of an LRS spacetime by using equation (\ref{scalarcons}).
Inserting the scalar variables $\phi$, $\bra{\Sigma-\sfr23\Theta}$ and $\hh$, and using the field equations, we get the following set of equations:
\ba 
\Omega Q-2\xi^3+2\Omega^2\xi-\sfr13\xi\Theta\Sigma+\xi\Sigma^2+\sfr23\xi\mu && \nonumber \\ 
-\sfr29\xi\Theta^2 +\xi \E+\sfr12\xi\Pi+A\phi\xi &=& 0, \\
-\Omega A\phi-\sfr13\Omega\Theta\Sigma+\sfr29\Omega\Theta^2-2\Omega^3-\Omega E &&  \nonumber \\
 +\sfr12\Omega\Pi+\sfr13\Omega\mu+\Omega p + \xi\phi\Sigma+2\Omega\xi^2+Q\xi &=& 0, \\
\sfr92\Omega\xi\Sigma^2-3\Omega\xi\Theta\Sigma+\Omega^2Q+\Omega\xi\Pi+\sfr92\xi^2\phi\Sigma && \nonumber \\
\xi\Omega\mu+\xi\Omega p+Q\xi^2 &=& 0.
\ea
Solving the above system of equations for $\E, p$ and $\phi$, we get 
\ba
p &=& -\frac{\Omega^2Q+\xi\Omega\mu+\Omega\xi\Pi+Q\xi^2}{\Omega\xi}, \\
\phi &=& -\frac{\bra{3\Sigma-2\Theta}\Omega}{3\xi},
\ea
which are same as (\ref{omegaxi2}) and (\ref{constraint1}). Also we get a new algebraic relation for $\E$:
\ba\label{E}
\E &=&  \frac\Omega\xi \udot (\Sigma-\frac23\Theta)-\Sigma^2+\frac13\Theta\Sigma+\frac29\Theta^2\nonumber\\
&&+2(\xi^2-\Omega^2)- \frac\Omega\xi Q -\frac12\Pi-\frac23\mu\;.
\ea

The above equation along with equation (\ref{H}) completely describes the Weyl tensor in LRS spacetimes. 

Now taking into account the results for perfect fluid LRS spacetimes \cite{Elst_Ellis_1996}, we see that the above phenomenon is true for any LRS spacetime and we can state this interesting theorem:

\begin{thm}
The symmetry of LRS spacetimes makes the Weyl tensor obey an algebraic constraint with other 1+1+2 geometrical variables.  Hence the doubly contracted Bianchi identities that describe the propagation and evolution of the Weyl tensor become redundant. 
\end{thm}

Now we can see that the number of independent geometrical 1+1+2 scalars that describe a LRS spacetime is reduced considerably. For example, specifying \[ {\cal D}_3 :=\brac{\udot, \Theta, \xi, \Sigma, \mu, p(\mu,\Pi,Q), \Pi, Q }, \] will automatically specify \[ {\cal D}_4 :=\brac{\Omega, \phi, \E, \H} ,\] via the constraint equations (\ref{constraint1},\ref{constraint2},\ref{E},\ref{H}). Hence, we can give the initial Cauchy data on any spacelike Cauchy surface for the independent variables using any suitable chosen equation of state $p(\mu, \Pi, Q)$, and then determine their evolution via equation (\ref{waveeqn}), which applies equally to all the variables in ${\cal D}_3$. This will then provide us with a unique self-similar dynamical solution for the LRS spacetimes with non-zero rotation and spatial twist. The nature of the matter required for such solutions to exist follows from Theorem 2, where there is a non trivial condition on the presence of heat flux. There are no other conditions on the density, pressure or pressure anisotropy, but for a physically realistic solutions these must obey the Dominant energy conditions (DEC).

As described in detail in \cite{Ellis_1968}, if a spacetime exhibit local rotational symmetry in an open neighbourhood of a point $P$, then the coordinate freedoms can be used to describe the local metric in the neighbourhood in $(t,r,x,y)$ coordinates in the following way:
\ba\label{metric}
ds^2&=&-F^2(t,r)dt^2+X^2(t,r)dr^2\nonumber\\
&&+Y^2(t,r)[dx^2+D(x)dy^2]\nonumber\\
&&+g(x)F^2(t,r)[2dt-g(x)dy]dy\nonumber\\
&&-h(x)X^2(t,r)[2dr-h(x)dy]dy
\ea
We can immediately see that $g(x)=h(x)=0$ and $D(x)=\sin^2x$ gives a general spherically symmetric metric which is of LRS class II. However we have already established that LRS spacetimes with non-vanishing rotation and spatial twist must be self similar. Hence the functions $F$, $X$ and $Y$ can be written in terms of a single variable $z\equiv t/r$. Hence only self similar spherically symmetric solutions can be obtained in the limit $g(x)\rightarrow 0, h(x)\rightarrow 0$ or equivalently $\Omega\rightarrow 0,\xi\rightarrow 0$.
Therefore to study the interior of a rotating, radiating and inhomogeneous star as a first approximation from the spherical symmetry, we can start with a self similar spherically symmetric spacetime and add sufficiently small $g(x)$ and $h(x)$, with respect to some covariant scale in the problem (the Misner Sharp mass of the spherical star for example) and solve the field equations with the matter source that obeys all the restrictions as imposed by Theorem 2 and the energy conditions.

\section{Discussion}

In this paper we transparently showed that it is possible to have a Locally Rotationally Symmetric spacetime with non-zero rotation and spatial twist simultaneously if we allow for non-zero and bounded heat flux. We investigated in detail all the covariant geometrical properties of such spacetimes and proved an interesting result that evolution of all the covariant scalars obey a single common hyperbolic linear second order partial differential equation. The existence of spacelike Cauchy surface, where initial Cauchy data can be provided is guaranteed. It was also shown that these solutions are self similar as they possess a conformal Killing vector in the $[u,e]$ plane.\\

As these solutions are neither stationary nor spatially homogeneous in general, with suitable equations of state, perhaps with the temperature $T$ as an internal variable in the equations of state for $P$, $\Pi$, and $Q$, they have the potential to give exact general relativistic models for rotating and dynamic and radiating stellar structures as they definitely have non zero heat flux in the interior. These solutions will then provide a relativistic description of a rotating stellar interior 
with quadrupole and other higher multipole moments and this may account for physical features of stars that cannot be explained by Newtonian dynamics.

 \begin{acknowledgments}
SS. GFRE and RG are supported by National Research Foundation (NRF), South Africa. SDM 
acknowledges that this work is based on research supported by the South African Research Chair Initiative of the Department of
Science and Technology and the National Research Foundation.
\end{acknowledgments}

\end{document}